\documentclass{article}

\usepackage{latexsym,amsthm,amsmath,amssymb,url}

\newcommand{\pp}{\hat{\phi}}

\newcommand{\2}{\vspace{3mm}}

\newcommand{\Tt}[1]{\mbox{$\Delta(#1)$}}
\newcommand{\MST}{\textsc{Max-SAT}}

\newcommand{\maxw}{{\rm sat}}

\let\phi=\varphi

\newtheorem{lemma}{Lemma}

\newtheorem{theorem}{Theorem}

\begin{document}

 \title{A New Lower Bound on the Maximum Number of Satisfied Clauses in Max-SAT and its Algorithmic Applications
 \thanks{A preliminary version of this paper appeared in the Proceedings of IPEC 2010. Research was supported in part by an International Joint grant of Royal Society. Research of Gutin was also supported
 in part by the IST Programme of the European Community, under the
 PASCAL 2 Network of Excellence.}}

 \author{Robert Crowston, Gregory Gutin, Mark Jones, Anders Yeo\\
\small Department of Computer Science\\[-3pt]
\small  Royal Holloway, University of London\\[-3pt]
\small Egham, Surrey TW20 0EX, UK\\[-3pt]
\small \texttt{robert,gutin,markj,anders@cs.rhul.ac.uk}}

 \date{ }
 \maketitle

\begin{abstract}
A pair of unit clauses is called conflicting if it is of the form $(x)$, $(\bar{x})$. A CNF formula is unit-conflict free (UCF) if it contains no pair
of conflicting unit clauses. Lieberherr and Specker (J. ACM 28, 1981) showed that for each UCF CNF formula with $m$ clauses we can simultaneously satisfy at least $\pp m$ clauses, where $\pp =(\sqrt{5}-1)/2$. We improve the Lieberherr-Specker bound by showing that for each UCF CNF formula $F$ with $m$ clauses we can find, in polynomial time, a subformula $F'$ with $m'$ clauses such that we can simultaneously satisfy at least $\pp m+(1-\pp)m'+(2-3\pp)n''/2$ clauses (in $F$), where
$n''$ is the number of variables in $F$ which are not in $F'$.

We consider two parameterized versions of MAX-SAT, where the parameter is the number of satisfied clauses above the bounds $m/2$ and $m(\sqrt{5}-1)/2$. The former bound is tight for general formulas, and the later is tight for UCF formulas. Mahajan and Raman (J. Algorithms 31, 1999) showed that every instance of the first parameterized problem can be transformed, in polynomial time, into an equivalent one with at most $6k+3$ variables and $10k$ clauses. We improve this to $4k$ variables and $(2\sqrt{5}+4)k$ clauses. Mahajan and Raman conjectured that the second parameterized problem is fixed-parameter tractable (FPT). We show that the problem is indeed FPT by describing a polynomial-time algorithm that transforms any problem instance into an equivalent one with at most $(7+3\sqrt{5})k$ variables. Our results are obtained using our improvement of the Lieberherr-Specker bound above.
\end{abstract}


\section{Introduction}\label{sec:Intro}

Let $F = (V,C)$ be a CNF formula, with a set $V$ of variables and a multiset $C$ of non-empty clauses, $m=|C|$ (i.e., $m$ is the number of clauses in $C$; each clause is counted as many times as it appears in $C$), and $\maxw(F)$ is the maximum
number of clauses that can be satisfied by a truth assignment.
With a random assignment of truth values to the variables, the probability of a clause being satisfied is at least $1/2$.
Thus,  $\maxw(F) \ge m/2$ for any $F$. This bound is tight when $F$ consists of pairs of {\em conflicting unit clauses} $(x)$ and $(\bar{x})$.
Since each truth assignment satisfies exactly one clause in each pair of conflicting unit clauses, it is natural to reduce
$F$ to the {\em unit-conflict free (UCF)} form
by deleting all pairs of conflicting clauses. If $F$ is UCF, then Lieberherr and Specker \cite{LieberherrSpecker81} proved
that ${\rm sat}(F)\ge \pp m$, where $\pp =(\sqrt{5}-1)/2$ (golden ratio inverse), and that for any $\epsilon > 0$ there are UCF CNF formulae $F$ for which
$\maxw(F) < m(\pp + \epsilon)$. Yannakakis \cite{Yannakakis94} gave a short probabilistic proof that ${\maxw(F)} \ge \pp m$ by showing
that if the probability of every variable appearing in a unit clause being assigned {\sc true} is $\pp$ (here we assume that for all
such variables $x$ the unit clauses are of the form $(x)$) and the probability of every other variable being assigned {\sc true}
is $1/2,$ then the expected number of satisfied clauses is $\pp m$.

A formula $F'=(V',C')$ is called a {\em subformula} of a CNF formula $F=(V,C)$ if $C'\subseteq C$ and $V'$ is the set of variables in $C'.$
If $F'$ is a subformula of $F$ then $F\setminus F'$ denotes the subformula obtained from $F$ by deleting all clauses of $F'$.
A formula $F=(V,C)$ is called {\em expanding} if for each $X\subseteq V$, the number of clauses containing at least one variable from $X$ is at least $|X|$ \cite{Sze2004}.
It is known (this involves so-called matching autarkies, see Section \ref{sec:NT} for details) that for each CNF formula $F=(V,C)$ a subformula $F'=(V',C')$  can be found in polynomial time such that ${\rm sat}(F)={\rm sat}(F\setminus F')+|C'|$ and the subformula
$F\setminus F'$ is expanding.
In this paper, the main technical result is that ${\rm sat}(F) \ge \pp |C| + (2-3\pp)|V|/2$ for every expanding UCF CNF formula $F = (V,C)$.
Combining this inequality with the previous equality for ${\rm sat}(F)$, we conclude that for each UCF  CNF formula $F=(V,C)$
a subformula $F'=(V',C')$ can be found in polynomial time such that $${\rm sat}(F)\ge \pp |C|+(1-\pp)|C'| + (2-3\pp)|V\setminus V'|/2.$$
The last inequality improves the Lieberherr-Specker lower bound on ${\rm sat}(F)$.

Mahajan and Raman \cite{Mahajan99} were the first to recognize both practical and theoretical importance of parameterizing maximization
problems above tight lower bounds. (We give some basic
terminology on parameterized algorithms and complexity in the next section.) They considered \MST{} parameterized above the tight lower bound  $m/2$:

\begin{quote}
  {\bfseries {\sc SAT-A($m/2$)}}\\
  \emph{Instance:} A CNF formula $F$ with $m$ clauses.\\
  \emph{Parameter:} A nonnegative integer~$k$.\\
  \emph{Question:} Decide whether ${\rm sat}(F)\ge m/2 + k.$
\end{quote}

Mahajan and Raman proved that {\sc SAT-A($m/2$)} is fixed-parameter tractable by obtaining a problem kernel with at most $6k+3$ variables
and $10k$ clauses. We improve on this by obtaining a kernel with at most $4k$ variables and $8.47k$ clauses.

Since $\pp m$ rather than $m/2$ is an asymptotically tight lower bound for UCF CNF formulae, Mahajan and Raman \cite{Mahajan99}
also introduced the following parameterization of \MST{}:

\begin{quote}
  {\bfseries {\sc SAT-A($\pp m$)}}\\
  \emph{Instance:} A UCF CNF formula $F$ with $m$ clauses.\\
  \emph{Parameter:} A nonnegative integer~$k$.\\
  \emph{Question:} Decide whether ${\rm sat}(F)\ge \pp m + k.$
\end{quote}

Mahajan and Raman conjectured that {\sc SAT-A($\pp m$)}
is fixed-parameter tractable. To solve the conjecture in the affirmative, we show the existence of an $O(k)$-variable kernel
for {\sc SAT-A($\pp m$)}. This result follows from
our improvement of the Lieberherr-Specker lower bound.

The rest of this paper is organized as follows.  In Section \ref{sec:NT}, we give further terminology and notation and some basic results.
Section \ref{sec:MR} proves the improvement of the Lieberherr-Specker lower bound on ${\rm sat}(F)$ assuming correctness of
the following lemma: if $F=(V,C)$ is a compact CNF formula,  then $\maxw(F) \ge \pp |C| + (2-3\pp)|V|/2$
(we give definition of a compact CNF formula in the next section). We prove this non-trivial lemma in Section \ref{sec:ML}.
In Section \ref{sec:fpt} we solve the conjecture of Mahajan and Raman \cite{Mahajan99} in the affirmative and improve their result on {\sc SAT-A($m/2$)}.
We conclude the paper with discussions and open problems.

\section{Additional Terminology, Notation and Basic Results}\label{sec:NT}

We let $F = (V,C)$ denote a CNF formula with a set of variables $V$ and a multiset of clauses $C$.
It is normally assumed that each clause may appear multiple times in $C$. For the sake of convenience, we assume that each clause appears at
most once, but allow each clause to have an integer \emph{weight}. (Thus, instead of saying a clause $c$ appears $t$ times, we will say
that $c$ has weight $t$). If at any point a particular clause $c$ appears more than once in $C$, we replace all occurrences of $c$ with a
single occurrence of the same total weight.
We use $w(c)$ to denote the weight of a clause $c$. For any clause $c \notin C$ we set $w(c)=0$.
If $C' \subseteq C$ is a subset of clauses, then $w(C')$ denotes the sum of the weights of the clauses in $C'$. For a formula $F=(V,C)$
we will often write $w(F)$ instead of $w(C)$.

For  a formula $F=(V,C)$ and a subset $U\subseteq V$ of variables, $F_U$ denotes the subformula of $F$ obtained from $F$ by deleting all clauses without variables in $U$.

For a CNF formula $F=(V,C)$, a {\em truth assignment} is a function $\alpha : V \rightarrow \{ \textsc{true, false} \}$.
A truth assignment $\alpha$ \emph{satisfies} a clause $c$ if there exists $x \in V$ such that $x \in c$ and
$\alpha(x)=$ {\sc true}, or $\bar{x} \in c$ and $\alpha(x)=$  {\sc false}.  The {\em weight}
of a truth assignment is the sum of the weights of all clauses satisfied by the assignment. The maximum weight of
a truth assignment for $F$ is denoted by $\maxw(F)$.

A function $\beta : U \rightarrow \{ \textsc{true, false} \}$, where $U$ is a subset of $V$ is called a {\em partial truth assignment}.
A partial truth assignment $\beta : U \rightarrow \{ \textsc{true, false} \}$ is an {\em autarky} if $\beta$ satisfies all clauses of $F_U$. Autarkies are of interest, in particular, due to the following simple fact whose trivial proof is omitted.

\begin{lemma}\label{lem:aut}
Let $\beta : U \rightarrow \{ \textsc{true, false} \}$ be an autarky for a CNF formula $F$.
Then ${\rm sat}(F)=w(F_U)+{\rm sat}(F\setminus F_U).$
\end{lemma}

A version of Lemma \ref{lem:aut} can be traced back to Monien and Speckenmeyer \cite{MS1985}.

Recall that a formula $F=(V,C)$ is called {\em expanding} if $|X|\le w(F_X)$ for each $X\subseteq V$.
We associate a bipartite graph with a CNF formula $F=(V,C)$ as follows: the {\em bipartite graph} $B_F$ {\em of $F$}
has partite sets $V$ and $C$ and the edge $vc$ is in $B_F$ if and only if the variable $v$ or its negation $\bar{v}$ appears in the clause $c$.
Later we will make use of the following result which is a version of Hall's Theorem on matchings in bipartite graphs (cf. \cite{West}).

\begin{lemma}\label{lem:HT}
The bipartite graph $B_F$ has a matching covering $V$ if and only if $F$ is expanding.
\end{lemma}

We call a CNF formula $F= (V,C)$ {\em compact} if the following conditions hold:

\begin{enumerate}
 \item All clauses in $F$ have the form $(x)$ or $(\bar{x} \vee \bar{y})$ for some $x,y \in V$.
\item For every variable $x \in V$, the clause  $(x)$ is in $C$.
\end{enumerate}

A \emph{parameterized problem} is a subset $L\subseteq \Sigma^* \times
\mathbb{N}$ over a finite alphabet $\Sigma$. $L$ is
\emph{fixed-parameter tractable} if the membership of an instance
$(I,k)$ in $\Sigma^* \times \mathbb{N}$ can be decided in time
$f(k)|I|^{O(1)}$ where $f$ is a function of the
{\em parameter} $k$ only~\cite{DowneyFellows99,FlumGrohe06,Niedermeier06}.
Given a parameterized problem $L$,
a \emph{kernelization of $L$} is a polynomial-time
algorithm that maps an instance $(x,k)$ to an instance $(x',k')$ (the
\emph{kernel}) such that (i)~$(x,k)\in L$ if and only if
$(x',k')\in L$, (ii)~ $k'\leq h(k)$, and (iii)~$|x'|\leq g(k)$ for some
functions $h$ and $g$.
It is well-known \cite{DowneyFellows99,FlumGrohe06,Niedermeier06} that a decidable parameterized problem $L$ is fixed-parameter
tractable if and only if it has a kernel. By replacing  Condition (ii) in the definition of a kernel  by $k'\le k$,
we obtain a definition of a {\em proper kernel} (sometimes, it is called a {\em strong kernel}); cf. \cite{AF,CFM}.

\section{New Lower Bound for $\maxw(F)$}\label{sec:MR}

We would like to prove a lower bound on ${\rm sat}(F)$ that includes the number of variables as a factor. It is clear that for general CNF formula $F$ such a bound is impossible. For consider a formula containing a single clause $c$ containing a large number of variables. We can arbitrarily increase the number of variables in the formula, and the maximum number of satisfiable clauses will always be 1.
We therefore need a reduction rule that cuts out 'excess' variables. Our reduction rule is based on the following lemma proved in Fleischner et al. \cite{FKS2002} (Lemma 10) and Szeider \cite{Sze2004} (Lemma 9).

\begin{lemma}\label{lem:red}
Let $F=(V,C)$ be a CNF formula.
Given a maximum matching in the bipartite graph $B_F$, in time $O(|C|)$ we can find an autarky $\beta : U \rightarrow \{ \textsc{true, false} \}$ such that
$|X|+1 \le w(F_X)$ for
every $X \subseteq V \backslash U.$
\end{lemma}

Note that the autarky found in Lemma \ref{lem:red} can be empty, i.e., $U=\emptyset$.
An autarky found by the algorithm of Lemma \ref{lem:red} is of a special kind, called a matching autarky; such autarkies were used first by Aharoni and Linial \cite{AL1986}. Results similar to Lemma \ref{lem:red} have been obtained in the parameterized complexity literature as well, see, e.g., \cite{LokshtanovPhD}.

Lemmas \ref{lem:aut} and \ref{lem:red} immediately imply the following:

\begin{theorem} \label{match}\cite{FKS2002,Sze2004}
Let $F$ be a CNF formula and let $\beta : U \rightarrow \{ \textsc{true, false} \}$ be an autarky found by the algorithm of Lemma \ref{lem:aut}.
Then ${\rm sat}(F)={\rm sat}(F\setminus F_U)+w(F_U)$ and $F\setminus F_U$ is an expanding formula.
\end{theorem}

Our improvement of the Lieberherr-Specker lower bound on $\maxw(F)$ for a UCF CNF formula $F$ will follow immediately from Theorems \ref{match} and \ref{main} (stated below). It is much harder to prove Theorem \ref{main} than Theorem \ref{match}, and our proof of Theorem \ref{main} is based on the following quite non-trivial lemma that will be proved in the next section.

\begin{lemma}\label{compact}
  If $F=(V,C)$ is a compact CNF formula, then there exists a truth assignment with weight at least
$$\pp w(C) + \frac{|V|(2-3\pp)}{2},$$
where $\pp=(\sqrt{5}-1)/2$, and such an assignment can be found in polynomial time.
\end{lemma}

The next proof builds on some of the basic ideas in \cite{LieberherrSpecker81}.

\begin{theorem}\label{main}
  If $F=(V,C)$ is an expanding UCF CNF formula, then there exists a truth assignment with weight at least
$$\pp w(C) + \frac{|V|(2-3\pp)}{2},$$
where $\pp=(\sqrt{5}-1)/2$ and such an assignment can be found in polynomial time.
\end{theorem}
\begin{proof}
We will describe a polynomial-time transformation from $F$ to a compact CNF formula $F'$, such that $|V'| = |V|$ and  $w(C') = w(C)$,
and any truth assignment for $F'$ can be turned into truth assignment for $F$ of greater or equal weight. The theorem then follows from Lemma \ref{compact}.

By Lemma \ref{lem:HT}, there is a matching in the bipartite graph $B_F$ covering $V.$
For each $x \in V$ let $c_x$ be the unique clause associated with $x$ in this matching.
For each variable $x$, if the unit clause $(x)$ or $(\bar{x})$ appears in $C$, leave $c_x$ as it is for now. Otherwise, remove all variables except
$x$ from $c_x$. We now have that for every $x$, exactly one of $(x)$, $(\bar{x})$ appears in $C$.

If $(\bar{x})$ is in $C,$ replace every occurrence of the literal $\bar{x}$ in the clauses of $C$ with $x$, and replace every occurrence of $x$ with $\bar{x}$. We now have that Condition 2 in the definition of a compact formula is satisfied.
For any clause $c$ which contains more than one variable and at least one positive literal, remove all variables except one that occurs as a positive. For any clause which contains only negative literals, remove all but two variables. We now have that Condition 1 is satisfied. This completes the transformation.

In the transformation, no clauses or variables were completely removed, so $|V'| = |V|$ and $w(C') = w(C)$.
Observe that the transformation takes polynomial time, and that any truth assignment for the compact formula $F'$ can be turned into a truth assignment for $F$ of greater or equal weight. Indeed, for some truth assignment for $F'$, flip the assignment to $x$ if and only if we replaced occurrences of $x$ with $\bar{x}$ in the transformation. This gives a truth assignment for $F$ such that every clause will be satisfied if its corresponding clause in $F'$ is satisfied.
\end{proof}

Our main result follows immediately from Theorems \ref{match} and \ref{main}.

\begin{theorem}\label{main1}
 Every UCF CNF formula $F=(V,C)$ contains a (possibly empty) subformula $F'=(V',C')$ that can be found in polynomial time and such that $${\rm sat}(F)\ge \pp w(C)+(1-\pp)w(C') + (2-3\pp)|V\setminus V'|/2.$$
\end{theorem}

\section{Proof of Lemma \ref{compact}}\label{sec:ML}

In this section, we use the fact that $\pp=(\sqrt{5}-1)/2$ is the positive root of the polynomial $\pp^2+\pp-1$.
We call a clause $(\bar{x} \vee \bar{y})$ \emph{good} if for every literal $\bar{z}$, the set of clauses containing $\bar{z}$ is not equal
to $\{(\bar{x} \vee \bar{z}),(\bar{y} \vee \bar{z})\}.$
We define $w_v(x)$ to be the total weight of all clauses containing the literal $x$, and $w_v(\bar{x})$ the total weight of all clauses
containing the literal $\bar{x}$. (Note that $w_v(\bar{x})$ is different from $w(\bar{x})$, which is the weight of the particular clause $(\bar{x})$.)
Let $\epsilon(x) = w_v(x) - \pp w_v(\bar{x})$. Let $\gamma = (2-3\pp)/2=(1-\pp)^2/2$ and let $\Tt{F}=\maxw(F)-\pp w(C)$.

To prove Lemma \ref{compact}, we will use an algorithm, Algorithm A, described below. We will show that, for any compact CNF formula $F=(V,C)$,
Algorithm A finds a truth assignment with weight at least $\pp w(C) + \gamma |V|$. Step 3 of the algorithm removes any clauses which are satisfied or
falsified by the given assignment of truth values to the variables. The purpose of Step 4 is to make sure the new formula is compact.

Algorithm A works as follows.
Let $F$ be a compact CNF formula. If $F$ contains a variable $x$ such that
we can assign $x$ {\sc true} and increase
sufficiently the average number of satisfied clauses, we do just that (see
Cases A and B of the algorithm). Otherwise,
to achieve similar effect we have to assign truth values to two or three
variables (see Cases C and D).
Step 3 of the algorithm removes any clauses which are satisfied or
falsified by the given assignment of truth values to the variables. The
purpose of Step 4 is to make sure the new formula is compact.

\begin{center}
\fbox{~\begin{minipage}{11cm}
\textsc{Algorithm A}

\smallskip

While $|V|>0$, repeat the following steps:
\begin{enumerate}
 \item For each $x \in V$, calculate $w_v(x)$ and $w_v(\bar{x})$.
\item Mark some of the variables as {\sc true} or {\sc false}, according to the following cases:

\begin{description}
\item[Case A:] \emph{There exists $x \in V$ with $w_v(x) \ge w_v(\bar{x})$.}
Pick one such $x$ and assign it  {\sc true}.

\item[Case B:] \emph{Case A is false, and there exists $x \in V$ with $(1-\pp)\epsilon(x) \geq \gamma$}.
Pick one such $x$ and assign it {\sc true}.

\item[Case C:] \emph{Cases A and B are false, and there exists a good clause.}
Pick such  a good clause $(\bar{x} \vee \bar{y})$, with (without loss of generality) $\epsilon(x) \ge \epsilon(y)$, and assign
$y$ {\sc false} and $x$  {\sc true}.

\item[Case D:]  \emph{Cases A, B and C are false.}
Pick any clause $(\bar{x} \vee \bar{y})$ and pick $z$ such that both clauses $(\bar{x} \vee \bar{z})$ and $(\bar{y} \vee \bar{z})$ exist.
Consider the six clauses $(x), (y), (z), (\bar{x} \vee \bar{y}), (\bar{x} \vee \bar{z}), (\bar{y} \vee \bar{z})$ and all $2^3$ 
assignments to the variables $x,y,z$, and pick an assignment maximizing the total weight of satisfied clauses among the six clauses. 
\end{description}

\item Perform the following simplification: For any variable $x$ assigned {\sc False}, remove any clause containing $\bar{x}$, remove the unit clause $(x)$, and remove $x$ from $V$.
For any variable $x$ assigned {\sc true}, remove the unit clause $(x)$, remove $\bar{x}$ from any clause containing $\bar{x}$ and remove $x$ from $V$.

\item For each $y$ remaining, if there is a clause of the form $(\bar{y})$, do the following: If the weight of this clause is greater than $w_v(y)$, then replace all clauses containing the variable $y$ (that is, literals $y$ or $\bar{y}$) with one clause $(y)$ of weight $w_v(\bar{y})-w_v(y)$. Otherwise remove $(\bar{y})$ from $C$ and change the weight of $(y)$ to $w(y)-w(\bar{y})$.

\end{enumerate}

\end{minipage}~}
\end{center}

In order to show that the algorithm finds a truth assignment with weight at least $\pp w(C) + \frac{|V|(2-3\pp)}{2}$, we need the following two lemmas.

\2

\begin{lemma}\label{lemTrue}
 For a formula $F$, if we assign a variable $x$ {\sc true}, and run Steps 3 and 4 of the algorithm, the resulting formula $F^*=(V^*,C^*)$
satisfies
\[
 \Tt{F} \ge \Tt{F^*} + (1-\pp) \epsilon(x).
\]
 Furthermore, we have $|V^*|=|V|-1$, unless there exists $y\in V^*$ such that $(y)$ and $(\bar{x} \vee \bar{y})$ are the only clauses containing
$y$ and they have the same weight. In this case, $y$ is removed from $V^*$.

\end{lemma}
\begin{proof}
Observe that at Step 3, the clause $(x)$ (of weight $w_v(x)$) is removed, clauses of the form $(\bar{x} \vee \bar{y})$ (total weight $w_v(\bar{x})$) become $(\bar{y})$, and the variable $x$ is removed from $V$.

At Step 4, observe that for each $y$ such that $(\bar{y})$ is now a clause, $w(C)$ is decreased by $2w_y$ and $\maxw(F)$ is decreased by $w_y$, where
$w_y = \min \{ w(y) , w(\bar{y})\} $.  Let $q = \sum_y{w_y}$, and observe that $q \le w_v(\bar{x})$. A variable $y$ will only be removed at this stage if the clause $(\bar{x} \vee \bar{y})$ was originally in $C$. We therefore have

\begin{enumerate}
 \item $\maxw(F^*) \le \maxw(F) - w_v(x) - q$
\item $w(C^*) = w(C) - w_v(x) - 2q$
\end{enumerate}

Using the above, we get

\begin{center}
$\begin{array}{rcl}
  \Tt{F} & = & \maxw(F)-\pp\cdot w(C) \\
  &\ge & (w_v(x) + \maxw(F^*) + q) - \pp(w(C^*) + 2q + w_v(x)) \\
  & = & \Tt{F^*} + (1-\pp)w_v(x) - (2\pp-1)q \\
  & \geq & \Tt{F^*} + (1-\pp)(\epsilon(x)+\pp\cdot w_v(\bar{x})) - (2\pp-1) w_v(\bar{x}) \\
  & = & \Tt{F^*} + (1-\pp-\pp^2) w_v(\bar{x}) + (1-\pp) \epsilon(x) \\
  & = & \Tt{F^*} + (1-\pp) \epsilon(x). \hspace*{\fill}
\end{array}$
\end{center}
 \end{proof}

\2

\begin{lemma}\label{lemFalse}
For a formula $F$, if we assign a variable $x$ {\sc False}, and run Steps 3 and 4 of the algorithm, the resulting formula
$F^{**}=(V^{**},C^{**})$ has $|V^{**}|=|V|-1$ and satisfies
$
 \Tt{F} \ge \Tt{F^{**}} - \pp \epsilon(x).
$
\end{lemma}
\begin{proof}
Observe that at Step 3, every clause containing the variable $x$ is removed, and no other clauses will be removed at Steps 3 and 4. Since the clause $(y)$ appears for every other variable $y$, this implies that $|V^{**}|=|V|-1$. We also have the following:  $\maxw(F^{**}) \le \maxw(F) - w_v(\bar{x})$ and $w(C^{**}) = w(C) - w_v(\bar{x}) - w_v(x).$
Thus,

\begin{center}
$\begin{array}{rcl}
  \Tt{F}  & = & \maxw(F)-\pp w(C) \\
  & \ge & (w_v(\bar{x}) + \maxw(F^{**})) - \pp(w(C^{**}) + w_v(\bar{x}) + w_v(x)) \\
  & = & \Tt{F^{**}} + (1-\pp)w_v(\bar{x}) - \pp w_v(x) \\
  & = & \Tt{F^{**}} + (1-\pp)w_v(\bar{x}) - \pp ( \epsilon(x)+\pp\cdot w_v(\bar{x}))  \\
  & = & \Tt{F^{**}} + (1-\pp-\pp^2)w_v(\bar{x}) - \pp \epsilon(x) \\
  & = & \Tt{F^{**}} - \pp \epsilon(x). \hspace*{\fill} \\
\end{array}$
\end{center}
 \end{proof}

Now we are ready to prove Lemma \ref{compact}.

\2

\noindent{\em Proof of Lemma \ref{compact}:}
We will show that Algorithm A finds a truth assignment with weight at least $\pp w(C) + \frac{|V|(2-3\pp)}{2}$. Note that the inequality
in the lemma can be reformulated as $\Tt{F} \geq \gamma |V|$.

Let $F$ and $\pp$ be defined as in the statement of the lemma.
Note that at each iteration of the algorithm, at least one variable is removed. Therefore, we will show the lemma by induction on $|V|$. If $|V|=0$ then
we are done trivially and if $|V|=1$ then we are done as $\maxw(F)=w(C) \geq \pp  w(C) + \gamma$ (as $w(C) \geq 1$).
So assume that $|V| \geq 2$.

For the induction step, let $F'=(V',C')$ be the formula resulting from $F$ after running Steps 1-4 of the algorithm, and assume that
$\Tt{F'} \geq \gamma |V'|$.
We show that $\Tt{F} \geq \gamma |V|$, by analyzing each possible case in Step 2 separately.

\2

{\bf Case A:} {\em $w_v(x) \geq w_v(\bar{x})$ for some $x \in V$.} In this case we let $x$ be {\sc true}, which
by Lemma \ref{lemTrue} implies the following:

\begin{center}
$\begin{array}{rcl}
  \Tt{F} & \geq &
                        \Tt{F'} + (1-\pp) \epsilon(x) \\
  & = & \Tt{F'} + (1-\pp) ( w_v(x) - \pp w_v(\bar{x})) \\
  & \geq & \Tt{F'} + (1-\pp) ( w_v(x) - \pp w_v(x)) \\
  & = & \Tt{F'} + (1-\pp)^2 w_v(x) \\
  & = & \Tt{F'} +2 \gamma w_v(x). \\
\end{array}$
\end{center}

  If $y \in V \setminus V'$, then either $y=x$ or $(\bar{x} \vee \bar{y}) \in C$. Therefore
$|V|-|V'| \leq w_v(\bar{x})+1 \leq w_v(x)+1$. As $w_v(x) \geq 1$ we note that $2 \gamma w_v(x) \geq
\gamma (w_v(x) +1)$. This implies the following, by induction, which completes the proof of Case A.

\begin{center}
$\begin{array}{rcl}
  \Tt{F} & \geq & \Tt{F'} + \gamma (w_v(x) +1) \\
 & \geq & \gamma |V'| + \gamma (w_v(x) +1)
 \geq  \gamma |V|. \\
\end{array}$
\end{center}

{\bf Case B:} {\em Case A is false, and $(1-\pp)\epsilon(x) \geq \gamma$ for some $x \in V$. }

 Again we let $x$ be {\sc true}. Since
$w_v(y) < w_v(\bar{y})$ for all $y \in V$, we have $|V|=|V'|+1$. Analogously
to Case A, using Lemma \ref{lemTrue}, we get the following:

\begin{center}
$\begin{array}{rcl}
  \Tt{F} & \geq & \Tt{F'} + (1-\pp) \epsilon(x) \\
 & \geq & \gamma |V'| + \gamma
  =   \gamma |V|. \\
\end{array}$
\end{center}

For Cases C and D, we generate a graph $G$ from the set of clauses.
The vertex set of $G$ is the variables in $V$ (i.e. $V(G)=V$) and there is an edge between $x$ and $y$
if and only if the clause $(\bar{x} \vee \bar{y})$ exists in $C$. A {\em good} edge in $G$ is an edge
$uv \in E(G)$ such that no vertex $z \in V$ has $N(z)=\{u,v\}$ (that is, an edge is good if and only if
the corresponding clause is good).

{\bf Case C:} {\em Cases A and B are false, and there exists a good clause $(\bar{x} \vee \bar{y})$.} Without loss of generality assume
that $\epsilon(x)
\geq \epsilon(y)$. We will first let $y$ be {\sc False} and then we will let $x$ be {\sc true}.  By letting $y$
be {\sc False} we get the following by Lemma \ref{lemFalse}, where $F^{**}$ is defined in Lemma \ref{lemFalse}:
$
  \Tt{F}\ge  \Tt{F^{**}} - \pp \epsilon(x).
$

Note that the clause $(\bar{x} \vee \bar{y})$ has been removed so $w^{**}_v(\bar{x}) = w_v(\bar{x}) -
w(\bar{x} \vee \bar{y})$ and $w^{**}_v(x) = w_v(x)$ (where $w^{**}_v(.)$ denote the weights in $F^{**}$).  Therefore using Lemma
\ref{lemTrue} on $F^{**}$ instead of $F$
we get the following, where the formula $F^*$ in Lemma \ref{lemTrue} is denoted by $F'$ below and $w^0=w(\bar{x} \vee \bar{y})$:

\begin{center}
$\begin{array}{rcl}
  \Tt{F^{**}} & \geq &
  \Tt{F'} + (1-\pp) ( w_v(x) - \pp(w_v(\bar{x}) -  w^0) ). \\
\end{array}$
\end{center}

First we show that $|V'|=|V^{**}|-1 = |V|-2$. Assume that $z \in V \setminus (V' \cup \{x,y\})$
and note that $N(z) \subseteq \{x,y\}$. Clearly $|N(z)|=1$ as $xy$ is a good edge.
If $N(z)=\{y\}$ then $(z) \in C'$, so we must have $N(z) = \{x\}$. However the only
way $z \not\in V'$ is if $w_v(z) = w_v(\bar{z})$, a contradiction as Case A is false. Therefore,
$|V'|= |V|-2$, and the following holds by the induction hypothesis.

\begin{center}
$\begin{array}{rcl}
  \Tt{F} & \ge & \Tt{F^{**}}
                             - \pp \epsilon(x) \\
  & \geq &  \Tt{F'} + (1-\pp) ( w_v(x) - \pp(w_v(\bar{x}) -  w^0) ) - \pp \epsilon(x) \\
  & \geq & \gamma |V'| +  (1-\pp) (\epsilon(x) + \pp w^0) - \pp \epsilon(x) \\
 & =  & \gamma |V| - 2 \gamma  +  (1-\pp)\pp w^0 - (2\pp-1)\epsilon(x). \\
\end{array}$
\end{center}

 We would be done if we can show that $2 \gamma  \leq  (1-\pp)\pp w^0 - (2\pp-1)\epsilon(x) $.
As $w^0 \geq 1$ and we know that, since Case B does not hold, $(1-\pp)\epsilon(x) < \gamma$, we would be done if
we can show that $2 \gamma  \leq  (1-\pp)\pp - (2\pp-1)\gamma / (1-\pp) $.
This is equivalent to $\gamma=(1-\pp)^2/2 \le \pp (1-\pp)^2$, which is true, completing the
proof of Case C.

\2

{\bf Case D:} {\em Cases A, B and C are false.} Then $G$ has no good edge.

Assume $xy$ is some edge in $G$ and $z \in V$
such that $N(z)=\{x,y\}$. As $xz$ is not a good edge there exists a $v \in V$, such that $N(v)=\{x,z\}$.
However $v$ is adjacent to $z$ and, thus, $v \in N(z) = \{x,y\}$, which implies that $v = y$.
This shows that $N(y)=\{x,z\}$. Analogously we can show that $N(x)=\{y,z\}$.   Therefore, the only clauses in
$C$ that contain a variable from $\{x,y,z\}$ form the following set:
$S=\{ (x), (y), (z), (\bar{x} \vee \bar{y}), (\bar{x} \vee \bar{z}), (\bar{y} \vee \bar{z})\}. $

Let $F'$ be the formula obtained by deleting the variables $x$, $y$ and $z$ and all clauses containing them.
Now consider the three assignments of truth values to $x,y,z$ such that only one of the three variables is assigned {\sc False}.
Observe that the total weight of clauses satisfied by these three assignments equals $$w_v(\bar{x})+w_v(\bar{y})+w_v(\bar{z})+2(w(x)+w(y)+w(z))=2W,$$
where $W$ is the total weight of the clauses in $S$. Thus, one of the three assignments satisfies the weight of at least $2W/3$ among the
clauses in $S$.
Observe also that $w(C)-w(C') \ge 6$, and, thus, the following holds.

\begin{center}
$\begin{array}{rcl}
  \Tt{F} & \geq & 2(w(C)-w(C'))/3 + \maxw(F') - \pp (w(C') +w(C) - w(C')) \\
 & \geq & \gamma |V'| + 2(w(C)-w(C'))/3 - \pp(w(C) - w(C')) \\
 & = & \gamma |V| - 3\gamma +(2-3\pp)(w(C) - w(C'))/3 \\
 & \geq & \gamma |V| - 3\gamma +2(2-3\pp)
  >  \gamma |V|. \\
\end{array}$
\end{center}

\2

This completes the proof of the correctness of Algorithm A. It remains to show that Algorithm A takes polynomial time.

Each iteration of the algorithm takes $O(nm)$ time. The algorithm stops when $V$ is empty, and at each iteration some variables are
removed from $V$. Therefore, the algorithm goes through at most $n$ iterations and, in total, it takes $O(n^2m)$ time.
This completes the proof of Lemma \ref{compact}.
\qed

\2

Note that the bound $(2-3\pp)/2$ in Lemma \ref{compact} cannot be improved due to the following
example. Let $l$ be any positive integer and let $F=(V,C)$ be defined such that
$V=\{x_1,x_2,\ldots,x_l,y_1,y_2,\ldots,y_l\}$ and $C$ contain the constraints
$(x_1)$, $(x_2)$,\ldots , $(x_l)$, $(y_1)$, $(y_2)$,\ldots , $(y_l)$ and $(\bar{x}_1 \vee \bar{y}_1)$,
$(\bar{x}_2 \vee \bar{y}_2)$,\ldots , $(\bar{x}_l \vee \bar{y}_l)$.
  Let the weight of every constraint be one and note that for every $i$ we can only satisfy
two of the three constraints $(x_i)$, $(y_i)$ and $(\bar{x}_i \vee \bar{y}_i)$.
Therefore $\maxw(F) = 2l$ and the following holds:

\begin{center}
$\begin{array}{c}
\pp  w(C) + \frac{|V|(2-3\pp)}{2} = 3l\pp + \frac{2l(2-3\pp)}{2} = l (3\pp + 2-3\pp) = 2l = \maxw(F). \\
\end{array}$
\end{center}

\section{Parameterized Complexity Results}\label{sec:fpt}

Recall that formulations of parameterized problems {\sc SAT-A($m/2$)} and {\sc SAT-A($\pp m$)} were given in Section \ref{sec:Intro}.

\begin{theorem} \label{reduce}
The problem {\sc SAT-A($\pp m$)} has a proper kernel with at most $\lfloor (7+3\sqrt{5})k\rfloor$ variables.
\end{theorem}
\begin{proof}
Consider an instance $(F=(V,C),k)$ of the problem. By Theorem \ref{match},
there is an autarky $\beta : U \rightarrow \{ \textsc{true, false} \}$ which can be found by the polynomial algorithm of Lemma \ref{lem:aut} such that
${\rm sat}(F)={\rm sat}(F\setminus F_U)+w(F_U)$ and $F\setminus F_U$ is an expanding formula.

If $U=V$, then ${\rm sat}(F)=w(F)$, and the kernel is trivial.

Now suppose that $U\neq V$ and denote $F\setminus F_U$ by $F'=(V',C').$  We want to choose an integral
parameter $k'$ such that $(F, k)$ is a {\sc yes}-instance of
the problem if and only if $(F', k')$ is a {\sc yes}-instance of the problem.
It is enough for $k'$ to satisfy  $\maxw(F) - \lfloor \pp w(F)\rfloor -k = \maxw(F') - \lfloor \pp w(F')\rfloor - k'$.
By Theorem \ref{match}, $\maxw(F') = \maxw(F) - w(F)+w(F')$.
Therefore, we can set $k'= k - w(F) + w(F')+ \lfloor \pp w(F)\rfloor - \lfloor \pp w(F')\rfloor$. Since $w(F)-w(F')\ge \lceil
\pp(w(F)-w(F'))\rceil\ge \lfloor \pp w(F)\rfloor - \lfloor \pp w(F')\rfloor$, we have $k'\le k.$

By Theorem \ref{main}, if $k' \le \frac{|V'|(2-3\pp)}{2}$, then $F$ is a {\sc yes}-instance of the problem.
Otherwise, $|V'| < \frac{2k}{2-3\pp}=(7+3\sqrt{5})k$. Note that $F'$ is not necessarily a kernel as $w(F')$ is not necessarily bounded by a
function of $k$. However, if $w(F')\ge 2^{2k/(2-3\pp)}$ then we can solve the instance $(F',k')$ in time $O(w(F')^2)$ and, thus, we may assume
that $w(F')< 2^{2k/(2-3\pp)}$, in which case, $F'$ is the required kernel.
\end{proof}

\2

\begin{theorem}
The problem {\sc SAT-A($m/2$)} has a proper kernel with at most $4k$ variables and $(2\sqrt{5}+4)k\leq 8.473k$ clauses.
\end{theorem}

\begin{proof}
First, we reduce the instance to a UCF instance $F=(V,C)$. As in Theorem \ref{reduce}, in polynomial time, we can obtain an expanding formula $F'=(V',C')$.
Again, we want to choose a parameter $k'$ such that $(F, k)$ is a {\sc yes}-instance if
and only if $(F', k')$ is a {\sc yes}-instance.

It is enough for $k'$ to satisfy
$\maxw(F) - \lfloor w(F)/2\rfloor - k  = \maxw(F') -  \lfloor w(F')/2\rfloor - k'$.
By Theorem \ref{match}, $\maxw(F') = \maxw(F) - w(F)+w(F')$.
Therefore, we can set $k'= k - \lceil w(F)/2 \rceil + \lceil w(F')/2 \rceil$. As $w(F') \le w(F)$, we have
$k'\le k.$

 By Theorem \ref{main}, there is a truth assignment for $F'$ with weight at least $\pp w(F') + \frac{|V'|(2-3\pp)}{2}$.
 Hence, if $k' \le (\pp-1/2) w(F') + \frac{|V'|(2-3\pp)}{2}$, the instance is a {\sc yes}-instance.
Otherwise,
\begin{equation}\label{eq:1}
k'-\frac{|V'|(2-3\pp)}{2}>(\pp-\frac{1}{2})w(F').
\end{equation}

The weaker bound $k'>(\pp-\frac{1}{2})w(F')$ is enough to give us the claimed bound on the total weight (i.e., the number) of clauses.
To bound the number of variables, note that since $F'$ is expanding, we can satisfy at
least $|V'|$ clauses. Thus, if $w(F')/2+k'\le |V'|$, the instance is a {\sc yes}-instance.
Otherwise, $w(F')/2+k' > |V'|$ and
\begin{equation}\label{eq:2}2(\pp-\frac{1}{2})(|V'|-k') < (\pp-\frac{1}{2})w(F').\end{equation}

Combining  Inequalities (\ref{eq:1}) and (\ref{eq:2}), we obtain:
$$2(\pp-\frac{1}{2})(|V'|-k') < (\pp-\frac{1}{2})w(F') < k' - \frac{|V'|(2-3\pp)}{2}.$$

This simplifies to $|V'|<4k'\le 4k$, giving the required kernel.
\end{proof}

\section{Discussion}

A CNF formula $I$ is $t$-{\em satisfiable} if any subset of $t$ clauses of $I$ can be satisfied simultaneously. In particular, a CNF formula is unit-conflict free if and only if it is 2-satisfiable.
Let $r_t$ be the largest real such that in any $t$-satisfiable CNF formula at least $r_t$-th fraction of its clauses can be satisfied simultaneously. Note that $r_1=1/2$ and $r_2=(\sqrt{5}-1)/2.$
Lieberherr and Specker \cite{LieberherrSpecker82} and, later, Yannakakis  \cite{Yannakakis94} proved that $r_3\ge 2/3$.  K{\"a}ppeli and Scheder \cite{KapScheder} proved that $r_3\le 2/3$ and, thus, $r_3 = 2/3$. Kr{\'a}l \cite{Kral} established the value of $r_4$: $r_4=3/(5+[(3\sqrt{69}-11)/2]^{1/3}-[3\sqrt{69}+11)/2]^{1/3})\approx 0.6992$.

For general $t$, Huang and Lieberherr \cite{HL85} showed that $\lim_{t\rightarrow \infty} r_t\le 3/4$ and Trevisan \cite{T1997} proved that $\lim_{t\rightarrow \infty} r_t=3/4$ (a different proof of this result is later given by Kr{\'a}l \cite{Kral}).

In the preliminary version of this paper published in the proceedings of IPEC 2010 we asked to establish parameterized complexity of the following parameterized problem: given a 3-satisfiable CNF formula $F=(V,C)$, decide whether $\maxw(F) \geq 2|C|/3+k$, where $k$ is the parameter.
This question was recently solved in \cite{GutJonYeo} by showing that the problem has a kernel with a linear number of variables. Unlike this paper, \cite{GutJonYeo} uses the Probabilistic Method. Similar question for any fixed $t>3$ remains open.

\end{document}